\documentstyle[11pt]{article}   % Stile del documento
\newtheorem{theorem}{Theorem}[section]

\newcommand{\qbox}{\hfill\rule{2mm}{2mm}}  % box pieno
\newenvironment{proof}{
\begin{trivlist}
\item[\hspace{\labelsep}{\bf\noindent Proof. }]
}{\qbox\end{trivlist}}

\begin{document}

\begin{center}
\vspace{0.1cm}
\Large
{\bf On the connection between orthant probabilities and
the first passage time problem} \\
\vspace{0.5cm}
\large
E. Di Nardo \\
\small
{\em Dipartimento di Matematica, Universit\`a degli Studi della Basilicata \\
E-mail: dinardo@unibas.it}
\end{center}
%
%\date{\empty}
%
\vspace{24pt}
\begin{abstract}
This article describes a new Monte Carlo method for the evaluation of the
orthant probabilities by sampling first passage times of a
non-singular Gaussian discrete time-series across an absorbing
boundary. This procedure makes use of a simulation of several
time-series sample paths, aiming to record their first crossing
instants. Thus, the computation of the orthant probabilities
is traced back to the accurate simulation of a non-singular
Gaussian discrete-time series. Moreover, if the simulation is also
efficient, this method is shown to be more speedy than the
others proposed in the literature. As example, we make use of the
Davies-Harte algorithm in the evaluation of the orthant probabilities
associated to the ARFIMA$(0,d,0)$ model. Test results are presented
that compare this method with currently available software.
\end{abstract}
{\bf Keywords:}
First passage time, orthant probabilities, simulation of
Gaussian discrete-time series, Davies-Harte algorithm.
% ================================================================================================
\section{Introduction}\label{section:1}
% ================================================================================================
For many computational problems in statistics and especially in economics
(cf. for example Ku and Seneta, 1994), we have to compute the
so-called orthant probabilities (Tong, 1990) of the form
\begin{equation}
P\left( \cap_{i=1}^{k} \left\{X_i<S_i \right\}
\right) = \int_{O_k} \frac{1}{(2\pi)^{k/2} |\Sigma|^{1/2}}
\exp\left\{-\frac{1}{2} {\bf x}^T  \Sigma^{-1} {\bf x} \right\} d{\bf
x},
\label{(notation)}
\end{equation}
where the random vector ${\bf X} \equiv (X_1,X_2, \ldots, X_k)$
follows a multivariate normal distribution with zero mean and unity
variance, $\Sigma$ is a $k \times k$ symmetric positive definite correlation
matrix and
$$O_k=\{{\bf x}\equiv(x_1,x_2,\ldots,x_k) \in {\bf R}^k: \,\,
x_i<S_i, \, S_i \in {\bf R}, \,\, i=1,2,...,k\}$$
are the so-called orthant regions. Since analytical closed forms
of the multidimensional integral (\ref{(notation)}) are known only
in a few special cases, their computation leads almost always to
numerical methods.
\par
Within numerical methods, the evaluation of (\ref{(notation)})
is full of history. The expansion in tetrachoric series
was the first attempt (cf. for a review Gupta, 1963), a method quickly abandoned
due to the slow convergence. A completely different
approach is the Plackett's formula (Plackett, 1954) consisting
of an integral dimension-reduction formula based on ad hoc conditioning.
This formula is used in the the routine G01HBF of the NAG Fortran software library
(see http://www.nag.co.uk/) that evaluates (\ref{(notation)}) up to $k=10.$
Different numerical algorithms can be simply obtained via standard numerical
integrations, such as the classical Monte-Carlo method or some adaptive quadrature
formulae. For $k=2$ and $k=3$ Genz (2001, see also references
therein) sets up a competitive algorithm based on Gaussian
quadratures with adaptive integration. If $k > 3,$ the infinite
integration limits in (\ref{(notation)}) need to be carefully
handled, either by using some type of transformation into a
finite region (see Genz, 1992 and 1993 and references therein)
or by using selected cutoff values. Genz has shown that the
quadrature formulae take less time and are more accurate than
the Monte Carlo algorithm for small $k,$ whereas the Monte Carlo
method appears more efficient and has comparable accuracy
with $k$ greater than $6$ or $7.$ In the last few years, a completely
different approach has spread, based on probability
simulation. In a review paper, Hajivassiliou {\it et al.} (1996)
analyze the properties of several simulators and find that the one they propose with the
label GHK (Geweke-Hajivassiliou-Keane) works better than all other
methods by keeping a good balance between accuracy and
computational costs. Here the idea consists of estimating (\ref{(notation)})
by means of recursive conditioned probabilities involving samples of truncated
standard normal random variables (r.v.'s). When $\Sigma$ has a
tridiagonal form or sparse structures, there are different numerical methods
available based on special decompositions of $\Sigma$ (for similar matters
see Tong, 1990). When the orthant regions have bounds like $S_i=0$ for
any $i,$ the interest has been focused on expressing (\ref{(notation)})
in terms of the correlation coefficients. Some solutions are available
in closed forms for $k=2$ and $k=3$ while for $k \geq 4$ integral
representations may be obtained. For example, Sun (1988) gives
analytical decompositions of (\ref{(notation)})
into a combination of several low order integrals.
A $k$-dimensional recursion formula is found in Zhongren and
Kedem (1999) by using a polar coordinates trasformation.
\par
In short, with up to $k=10$, there are different numerical methods
available by which orthant probabilities can be robustly and
reliably computed at low to moderate accuracy levels. High accuracy
or high dimension problems can require long computation
times and it is still not clear what is the best method for this type
of problem.
\par
The aim of this paper is to explore the connection between the
orthant probabilities (\ref{(notation)}) and the first passage
time (FPT) problem for a non-singular Gaussian discrete
time-series. This connection suggests a new way to
evaluate the orthant probabilities by numerical methods.
The keystone is choosing among some simulation procedures devoted to
the construction of time-series paths.
The mathematical kernel of the FPT problem for a non-singular
Gaussian time-series is recalled in section 2
(for a more detailed analysis see Di Nardo, 2002). The Genz
algorithm and the GHK simulator are summarized in section 3
together with some classical simulation procedures
of the time-series. As a working example, in the last section
the orthant probabilities are computed in connection with the
ARFIMA$(0,d,0)$ model by means of the FPT approach. Comparisons
are also given with the results reached via the Genz algorithm
and the GHK simulator.
%
%----------------------------------------------------------------
\section{The mathematical kernel}
%----------------------------------------------------------------
Let us define the FPT r.v. of a discrete time-series $X_t$ through an
absorbing time dependent boundary $S_t$ as
\begin{equation}
T=\min_{t \in N} \{t: X_t \geq S_t \}.
\label{(FPT)}
\end{equation}
In the following, suppose $P(X_0=x_0)=1$ with $x_0 < S_0$ and
in order to simplify set $x_0=0.$
\par
First we observe that
\begin{equation}
P(T=k)=P(X_1<S_1,X_2<S_2,\ldots,X_{k-1}<S_{k-1},X_k \geq S_k).
\label{(densFPT)}
\end{equation}
When $X_t$ is a Gaussian time-series with zero mean, unity
variance and symmetric positive definite correlation matrix $\Sigma,$
it is
$$P(T=k)=\int_{D_k} \frac{1}{(2\pi)^{k/2} |\Sigma|^{1/2}}
\exp\left\{-\frac{1}{2}{\bf x}^T  \Sigma^{-1} {\bf x} \right\} d{\bf x} $$
where
$$D_k=\{{\bf x} \equiv (x_1,x_2,\ldots,x_k) \in {\bf R}^k: \,\,
x_i<S_i, \, i=1,2,...,k-1,\, x_k \geq S_k\}.$$
Setting ${\bf S}\equiv(S_1,S_2,\ldots,S_k)$ and $P_k({\bf S}, \Sigma) =
P\left( \cap_{i=1}^{k} \left\{X_i<S_i \right\} \right),$ equation
(\ref{(densFPT)}) could be rewritten as
$$P(T=k)=\left\{ \begin{array}{ll}
1 - P_1({\bf S}, \Sigma), & \hbox{if $k=1$}, \\
P_{k-1}({\bf S}, \Sigma) - P_k({\bf S}, \Sigma), & \hbox{if $k >
1$}
\end{array} \right.$$
by which it results
\begin{equation}
P_{k}({\bf S}, \Sigma) = 1 - P(T\leq k).
\label{(distFPT)}
\end{equation}
The above equation provides the link between the orthant probabilities and the
FPT distribution of a non singular Gaussian time-series. This suggests the use
of the machinery of the FPT problems for a speedy computing of
(\ref{(notation)}).
\par
The next result allows us to estimate the FPT distribution
function via simulation of the time-series $X_t.$
\begin{theorem}
For an upper-bounded boundary $S_t,$ the FPT r.v. (\ref{(FPT)}) is
fair, i.e. $P(T<\infty)=1.$
\end{theorem}
\begin{proof}
Denote by $\rho_{ij}$ the correlation coefficients of $X_i$ and $X_j$ for
$i,j=1,2,\ldots,k.$ Set $S_{\max} = \max_{t \geq 0} S_t < \infty$ and $\rho_{max}=
\max_{i \ne j} \rho_{ij}.$ By using the Slepian inequality (Tong, 1990),
\begin{enumerate}
\item[{\it i)}] if $\rho_{max} > 0$ it is
\begin{equation}
P_k({\bf S}, \Sigma) \leq P_k({\bf S}, \Sigma_{\rho_{\max}})
\leq P_k({\bf S}_{\max}, \Sigma_{\rho_{\max}})
\label{(firstbound)}
\end{equation}
where $\Sigma_{\rho_{\max}}$ is the matrix with $1$
on the principal diagonal and $\rho_{\max}$ elsewhere and
${\bf S}_{\max}$ is a vector whose $k$ components are equal to
$S_{\max}.$ The random vector with normal distribution function
$N({\bf 0}, \Sigma_{\rho_{\max}})$ has exchangeable r.v.'s and so
(Tong, 1990)
$$P_k({\bf S}_{\max}, \Sigma_{\rho_{\max}}) = \int_{-\infty}^{\infty}
\Phi^k \left( \frac{S_{\max} + \sqrt{\rho_{\max}}
z}{\sqrt{1-\rho_{\max}}} \right) \phi(z) dz, \quad S \in {\bf R}$$
where
\begin{equation}
\phi(z) = \frac{1}{\sqrt{2 \pi}} \exp \left\{-\frac{z^2}{2} \right\},
\quad \Phi(x) =   \int_{-\infty}^ x \phi(z) dz, \quad x \in {\bf R};
\label{(distrNORM)}
\end{equation}
\item[{\it ii)}] if $\rho_{max} \leq 0$ it is
\begin{equation}
P_k({\bf S}, \Sigma) \leq P_k({\bf S}, I) \leq P_k({\bf S}_{\max}, I)
= \Phi^k(S_{\max})
\label{(secbound)}
\end{equation}
where $I$ is the identity matrix and $\Phi(x)$ is given in
(\ref{(distrNORM)}).
\end{enumerate}
From (\ref{(firstbound)}) and (\ref{(secbound)}), it is
\begin{equation}
\lim_{k \rightarrow \infty} P_k({\bf S}, \Sigma) = 0
\label{(zero)}
\end{equation}
and the result follows taking the limit in (\ref{(distFPT)})
for $k$ going to infinity.
\end{proof}
%
%-------------------------------------------------------------------------------
\section{Numerical evaluations of orthant probabilities}
%------------------------------------------------------------------------------
%
Let us recall that the analytical results on FPT problems are mostly centered on
stochastic processes of diffusion type, where the Markov
property plays a leading role in handling the related transition
probability density function (pdf), see Ricciardi {\it et al.}
(1999). The FPT distribution also has an explicit analytical expression
for the class of stochastic processes named the Levy type
anomalous diffusion, in which the mean square displacement of
the diffusive variable $X_t$ scales with time as $t^{\gamma}$
with $0 < \gamma < 2$ (see Balakrishnan 1985, Molchan 1999,
Rangarajan and Ding 2000). Nevertheless, apart from few special cases
no closed forms of the FPT distribution are available in the
literature. Thus numerical algorithms or simulation procedures have
been resorted to in order to get more information on FPT features. In
particular, simulation procedures are the tools mainly used, as
they are especially suitable to being implemented on parallel computers (see
Di Nardo {\it et al.}, 2000).
\par
A typical simulation procedure lies in sampling $N_S$ values of the FPT
r.v. $T,$ by a suitable construction of as many time-discrete
paths, and then to record the instants when such realizations
first cross the boundary.
\par
Note that $k$ consecutive observations of a standardized
Gaussian time-series can be generated via the innovations algorithm
from $k$ i.i.d. standard normal r.v.'s with $O(k^3)$
operations (Brockwell and Davies, 1991). However, if we add
few hypotheses on the correlation function, faster methods are
available. For example the Durbin-Levinson algorithm
(Brockwell and Davies, 1991) is an efficient procedure
with $O(k^2)$ operations if $X_t$ is stationary and
$\rho_{t}$ vanishing when $t$ goes to infinity.
If in addition the discrete Fourier transform of the circulant
embedding vector
\begin{equation}
\{\rho_0,\rho_1,\ldots, \rho_{k-1},\rho_k,\rho_{k-1},\ldots,\rho_1\}
\label{(circ)}
\end{equation}
is a positive real sequence, a still faster method is the Davies-Harte
algorithm (Davies {\it et al.}, 1987) with computational cost $O(k \log k)$
(for a review on the methods generating realizations of a Gaussian stationary
process see Percival, 1992).
\par
Therefore in order to evaluate the orthant probabilities (\ref{(notation)}), we suggest
the following Monte Carlo method: choose a faster method simulating the paths of $X_t,$
record their first crossing instants trough the boundary ${\bf S}$ and then estimate
(\ref{(notation)}) trough (\ref{(distFPT)}).
\par
In the next section, a FORTRAN 77 implementation of this procedure is compared
with the Genz algorithm and the GHK simulator that appear to be the
more widespread methods computing orthant probabilities.
\par
The Genz method (Genz, 1993) transforms the original integral
(\ref{(notation)}) into an integral over unit hypercube
\begin{equation}
P_k({\bf S}, \Sigma) = e_1 \int_0^1 e_2 \int_0^1 \ldots \int_0^1
e_k \int_0^1 d {\bf w},
\label{(eqGenz)}
\end{equation}
where $d {\bf w} \equiv (d w_1, \ldots, dw_k),$
$$e_i  = \left\{ \begin{array}{ll}
\Phi\left\{\displaystyle{\frac{S_1}{c_{11}}}\right\}, & \hbox{for $i=1$} \\
\Phi \left\{ \displaystyle{\frac{S_i - \sum_{j=1}^{i-1}
c_{ij} \Phi^{(-1)}(w_j e_j)}{c_{ii}}}\right\}, &  \hbox{for $i = 2, \cdots, k,$}
\end{array} \right.$$
$c_{ij}$ are elements of the lower triangular Cholesky decomposition of
$\Sigma,$ $\Phi(x)$ is given in (\ref{(distrNORM)}) and $\Phi^{(-1)}(x)$ is its inverse.
The idea of the Genz method is to apply a standard Monte Carlo method to
(\ref{(eqGenz)}) and evaluate $\{e_i\}$ by sampling pseudo-random
numbers over $(0,1).$ Let us observe that the Cholesky
decomposition of $\Sigma$ needs $O(k^3)$ computations for a fixed $k.$
The evaluation of the integrand function in (\ref{(eqGenz)})
takes $O(k^2)$ computations for each iteration and such an evaluation is repeated
until the estimated error is less than a prefixed tolerance or
the iterations number reaches a superior limit. The method has been implemented
in FORTRAN 77 and in GAUSS, and the routines are available at
http://www.sci.wsu.edu/math/faculty/genz/homepage up to $k=100.$
\par
The GHK simulator (Hajivassiliou {\it et al.}, 1996) estimates
(\ref{(notation)}) by means of
\begin{equation}
\tilde{P}_k({\bf S}, \Sigma) = \frac{1}{N} \sum_{n=1}^N
P(A_1)P(A_2|y_{1,n})\cdots P(A_k|y_{1,n}; y_{2,n}; \ldots ;y_{k-1,n})
\label{(GHK)}
\end{equation}
where $\{y_{i,n}\}$ are drawn sequentially from independent
standard normal distributions truncated to the events $A_i$
$$A_i=\left\{ -\infty < Y_i < \frac{S_i - \sum_{j=1}^{i-1} c_{ji}Y_j}{
c_{ii}} \right\}$$
with $c_{ij}$ elements of the lower triangular Cholesky
decomposition of $\Sigma$ and $\{Y_i\}_{i \in {\bf N}}$ a sequence of i.i.d.r.v.'s
having standard normal distribution. Again the Cholesky decomposition
of $\Sigma$ takes $O(k^3)$ computations for a fixed $k,$ whereas the evaluation of
(\ref{(GHK)}) is of order $O(N \times k^2)$ for $N$ sampling.
The method has been implemented in FORTRAN 77 and in GAUSS,
and the routines are available at http://econ.lse.ac.uk/staff/vassilis/
up to $k=40.$
%-------------------------------------------------------------------------------
\section{A working example}
%------------------------------------------------------------------------------
The Davies-Harte algorithm generates a path with $k+1$ steps of a stationary
Gaussian time-series having zero mean and autocovariances $\rho_0, \ldots, \rho_k,$
such that the finite Fourier transform of (\ref{(circ)}) is non-negative.
Craigmile (2003) has shown that the Davies-Harte algorithm is
the most efficient way to simulate an ARFIMA$(0,d,0)$ process
(Hosking, 1981). Let us recall that the standardized ARFIMA$(0,d,0)$
time-series is Gaussian, stationary with zero mean and correlation function
\begin{equation}
\rho_k=\frac{d(d+1)\cdots(d+k-1)}{(1-d)(2-d)\cdots(k-d)}, \quad k=1,2,\ldots
\label{(corrARFIMA)}
\end{equation}
if $\rho_0=1.$ To simulate a realization $X_1,X_2,\ldots,X_k$
with the Davies-Harte algorithm, the steps are:
\begin{enumerate}
\item[{\it i)}] for $n=0,1,2,\ldots, 2k-3$ evaluate
\begin{equation}
g_n = \sum_{j=0}^{k-2} \gamma(j) \exp \left( \frac{2 \pi i j n}{2 k - 2} \right)
+ \sum_{j=k-1}^{2k-3} \gamma(2 k - 2 - j) \exp \left( \frac{2 \pi i j n}{2 k - 2}
\right);
\label{(alg:1)}
\end{equation}
\item[{\it ii)}] for $j=0,1,\ldots,k-1$ evaluate
\begin{equation}
X_{j+1} = \frac{1}{2 \sqrt{k-1}} \sum_{n=0}^{2 k - 3} Z_n \sqrt{g_n} \exp \left(
\frac{2 \pi i j n}{2 k - 2}\right)
\label{(alg:2)}
\end{equation}
where $Z_0, Z_{k-1}$ are real normal r.v.'s with zero mean and variance
$2,$ $\{Z_n\}_{n=1}^{k-2}$ is a sequence of independent complex
normal r.v.'s with independent real and imaginary parts, each
of variance $1,$ and $Z_n=\overline{Z}_{2 k -2 - n}$ for $n=k, \ldots, 2 k -3.$
\end{enumerate}
Indeed, inverting equation (\ref{(alg:1)}) it is
$$\rho_j= \frac{1}{2k-2} \sum_{n=0}^{2k-3} g_n \exp\left( \frac{ 2 \pi i j n}{2k-2}
\right), j=0,1,\ldots,k-1$$
so that
\begin{eqnarray*}
Cov(X_{p+1} \, X_{q+1}) & = & E(X_{p+1} \, X_{q+1}) \\
& = & \frac{1}{4(k-1)} \sum_{n=0}^{2k-3} \sum_{l=0}^{2k-3} E(Z_n \overline{Z}_l)
\sqrt{g_n g_l} \exp \left(\frac{2 \pi i [p n - l q]}{2k-2} \right) \\
& = & \frac{1}{4(k-1)} \sum_{n=0}^{2k-3} 2 \sqrt{g_k^2} \exp \left( \frac{2 \pi i (p-q) k}{
2k-2}\right) = \rho_{p-q}.
\end{eqnarray*}
Then from (\ref{(distFPT)}), an estimation of (\ref{(notation)}) is
\begin{equation}
\tilde{P}_k({\bf S}, \Sigma) = 1 - \frac{\hbox{number of first crossings before or equal
$k$}}{N_S}
\label{(estimation)}
\end{equation}
where $N_S$ is the number of simulated paths.
\par
Observe that the step {\it i)} is
evaluated just once for any fixed $k$ with computational cost $O(k \log k),$
whereas the step {\it ii)} is repeated $N_{S}$ times each with $O(k \log k)$
evaluations. To speed up the simulation with the fast Fourier transform algorithm,
round $k$ up to the nearest power of two and truncate the simulated series at the
end.
\par
In the following, we take down some results in the evaluation of the orthant
probabilities with the Genz algorithm, the GHK simulator and the
method suggested here. Specifically, Table 1 refers to the case of
constant boundary $S(t)=1,$ $d=0.2$ and $k \geq 20$ and Table 2
refers to the case of linear boundary $S(t)=2 -0.01 \, t, $ $d=0.3$ and
$k \geq 20.$ In both cases, $\tilde{P}_k$ stands for numerical evaluations
of (\ref{(notation)}).
\par
With regard to the Genz method, the maximum number $N_{\max}$ of
integrand function evaluations has been allowed to $k \times 10^3,$ as suggested
by the same author. With this choice, the computed value $\tilde{P}_k$ has
reached the absolute accuracy $10^{-4}$ required in input, except in some cases marked
with star in the tables. To the $99\%$ confidence level, the estimated absolute
error has been near $10^{-4}$ for all $k \geq 20,$ except the cases
marked with star in the tables where it has been of order $10^{-3}.$ In Tables
$1 \div 2$, $CT_G$ represents an evaluation in seconds of the time employed
by the method in computing $\tilde{P}_k$ with such parameters.
\par
With regard to the GHK method, the number $N$ of allowed simulations
has been set equal to $N_S$ as indicated in Tables $1 \div 2.$ In order
to generate the required $N \times k$ pseudo-random uniform numbers
over $(0,1),$ it has been employed the routine G05CAF of the NAG software
library. In Tables $1 \div 2$, $CT_G$
represents an evaluation in seconds of the time employed by the
method in computing $\tilde{P}_k$ with such parameters.
\par
The method here proposed has been implemented in a
FORTRAN 77 program on a server ALPHA with operating system OSF1 5.0
and by using the NAG software library:
\begin{description}
\item[{\it i)}] the routine G05FDF generates a vector of normal
pseudo-random numbers by means of the Box and Muller method (see for example Rubinstein, 1981);
\item[{\it ii)}] the routine C06HBF computes the
discrete Fourier cosine transform of the sequence (\ref{(alg:1)});
\item[{\it iii)}] the routine C06EBF evaluates the discrete Fourier
transform of the Hermitian sequence in (\ref{(alg:2)}) with the fast
Fourier transform (see Brigham, 1974).
\end{description}
The time window size has been set equal to $2^5$ for $k=20 \div 30$
and equal to $2^6$ for $k=31 \div 40.$ In Tables $1 \div 2,$ $N_S$
is the number of simulated sample paths whereas $CT_{FPT}$
represents an evaluation in seconds of the time employed by the routine in
computing $\tilde{P}_k$ with such parameters.
\begin{table}[h]
\begin{center}
\begin{minipage}{.8\textwidth}
{\caption {\footnotesize Evaluations of the orthant probabilities $\tilde{P}_k$ for $S(t)=1$
and $d=0.2.$ $CT_G, CT_{GHK}$ and $CT_{FPT}$ represent an evaluation in seconds of the time
employed in the computation of $\tilde{P}_k$ respectively for the Genz method, the
GHK method and the FPT method. $N_S$ is the number of the simulated sample paths
in the FPT method.}}
\end{minipage} \\
\vspace{2mm}
\begin{tabular}[t]{cccccccc}
\hline
$ $ & \multicolumn{2}{c}{\small Genz method} & \multicolumn{2}{c} {\small GHK method} &
\multicolumn{3}{c}{\small FPT method} \\
\hline
$k$ & $\tilde{P}_k$ & $CT_G$ & $\tilde{P}_k$ & $CT_{GHK}$ & $\tilde{P}_k$ & $CT_{FPT}$ & $N_S$\\
\hline
20 & 0.0924  & 2.91 & 0.0927 & 0.49 & 0.0925 & 0.38 & 2000 \\
21 & 0.0835  & 7.19 & 0.0835 & 0.54 & 0.0838 & 0.41 & 2100 \\
22 & 0.0756  & 5.11 & 0.0757 & 0.58 & 0.0752 & 0.41 & 2100 \\
23 & 0.0684  & 8.07 & 0.0687 & 0.64 & 0.0686 & 0.42 & 2200 \\
24 & 0.0620  & 3.69 & 0.0622 & 0.68 & 0.0618 & 0.42 & 2200 \\
25 & 0.0563  & 8.98 & 0.0572 & 0.80 & 0.0558 & 0.46 & 2400 \\
26 & 0.0511  & 9.45 & 0.0512 & 0.93 & 0.0513 & 0.50 & 2650 \\
27 & 0.0463  & 4.28 & 0.0465 & 1.03 & 0.0460 & 0.50 & 2650 \\
28 & 0.0422  & 6.81 & 0.0422 & 1.05 & 0.0423 & 0.50 & 2650 \\
29 & 0.0383  & 7.18 & 0.0387 & 1.15 & 0.0385 & 0.52 & 2750 \\
30 & 0.0349  & 4.93 & 0.0348 & 1.25 & 0.0335 & 0.53 & 2800 \\
31 & 0.0317  & 3.33 & 0.0320 & 2.31 & 0.0318 & 1.43 & 3900 \\
32 & 0.0289  & 5.37 & 0.0291 & 2.45 & 0.0288 & 1.47 & 3950 \\
33 & 0.0264  & 5.57 & 0.0266 & 2.58 & 0.0262 & 1.48 & 4000 \\
34 & 0.0240  & 8.95 & 0.0243 & 2.76 & 0.0232 & 1.48 & 4000 \\
35 & 0.0220  & 6.03 & 0.0222 & 2.88 & 0.0215 & 1.48 & 4000 \\
36 & 0.0199  & 1.07 & 0.0202 & 3.01 & 0.0198 & 1.48 & 4000 \\
37 & 0.0182  & 6.64 & 0.0185 & 3.17 & 0.0182 & 2.30 & 6300 \\
38 & 0.0167  & 4.37 & 0.0167 & 5.39 & 0.0165 & 2.30 & 6300 \\
39 & 0.0153  & 6.95 & 0.0154 & 6.80 & 0.0155 & 2.34 & 6400 \\
40 & 0.0140  & 3.09 & 0.0140 & 6.18 & 0.0137 & 2.38 & 6500 \\
\hline
\end{tabular}
\end{center}
\end{table}

\begin{table}[h]
\begin{center}
\begin{minipage}{.8\textwidth}
{\caption {\footnotesize As in Table 1 for $S(t)=2-0.01 \, t$ and $d=0.3.$}}
\end{minipage} \\
\vspace{2mm}
\begin{tabular}[t]{clcccccc}
\hline
$ $ & \multicolumn{2}{c}{\small Genz method} & \multicolumn{2}{c} {\small GHK method} &
\multicolumn{3}{c}{\small FPT method} \\
\hline
$k$& $\tilde{P}_k$ & $CT_G$ & $\tilde{P}_k$ & $CT_{GHK}$ & $\tilde{P}_k$ & $CT_{FPT}$ & $N_S$\\
\hline
20 & 0.6661  & 25.00 & 0.6683 & 0.74 & 0.6661 & 0.62 & 3100 \\
21 & 0.6520* & 26.52 & 0.6523 & 0.93 & 0.6518 & 0.71 & 3500 \\
22 & 0.6381  & 12.43 & 0.6397 & 1.04 & 0.6381 & 0.75 & 3700 \\
23 & 0.6243  & 45.02 & 0.6247 & 1.25 & 0.6240 & 0.84 & 4200 \\
24 & 0.6107  & 9.16  & 0.6101 & 1.41 & 0.6106 & 0.90 & 4500 \\
25 & 0.5972  & 33.16 & 0.5966 & 1.52 & 0.5973 & 0.90 & 4500 \\
26 & 0.5838  & 53.18 & 0.5847 & 1.67 & 0.5832 & 0.93 & 4650 \\
27 & 0.5708  & 36.85 & 0.5711 & 1.80 & 0.5706 & 0.94 & 4700 \\
28 & 0.5578  & 25.49 & 0.5584 & 1.88 & 0.5578 & 0.94 & 4700 \\
29 & 0.5450  & 26.58 & 0.5456 & 2.00 & 0.5451 & 0.94 & 4700 \\
30 & 0.5323  & 31.18 & 0.5331 & 2.28 & 0.5324 & 0.99 & 4850 \\
31 & 0.5199  & 19.32 & 0.5207 & 2.96 & 0.5196 & 1.91 & 4950 \\
32 & 0.5075  & 22.96 & 0.5077 & 3.17 & 0.5080 & 1.96 & 5100 \\
33 & 0.4952  & 31.59 & 0.4967 & 3.47 & 0.4954 & 2.03 & 5200 \\
34 & 0.4833  & 32.85 & 0.4833 & 3.66 & 0.4847 & 2.03 & 5250 \\
35 & 0.4714  & 34.15 & 0.4725 & 4.04 & 0.4716 & 2.12 & 5500 \\
36 & 0.4596  & 35.32 & 0.4596 & 4.44 & 0.4596 & 2.23 & 5800 \\
37 & 0.4482  & 24.34 & 0.4482 & 4.70 & 0.4476 & 2.29 & 5900 \\
38 & 0.4368  & 38.15 & 0.4362 & 5.29 & 0.4357 & 2.37 & 6100 \\
39 & 0.4256  & 26.13 & 0.4249 & 5.73 & 0.4256 & 2.50 & 6400 \\
40 & 0.4146  & 40.59 & 0.4133 & 6.04 & 0.4146 & 2.52 & 6500 \\
\hline
\end{tabular}
\end{center}
\end{table}
As it would be theoretically expected, the GHK simulator and the Genz method
cost more in computational time than
the simulation procedure with the Davies-Harte algorithm.
Especially in the case of non-constant boundary, the Genz method is
time-consuming if one would reach the required accuracy.
Furthermore, the proposed Monte Carlo method is particularly suited
to be implemented on a parallel computer in order to further cut
the computational time and to improve the approximation;
this because the sample paths of the simulated time-series could be
easily generated independently of one another.

\end{document}